\documentclass{aamas2014}

\usepackage{cite}      
\usepackage{graphicx}  
\usepackage{caption}
\usepackage{subcaption}

\usepackage{amssymb}
\usepackage{mathrsfs}
\usepackage{amsmath}

\usepackage[sans]{dsfont}
\usepackage{stmaryrd}
\usepackage{pifont}
\usepackage{wasysym}
\usepackage{algorithm} 
\usepackage{array}
\usepackage{url}
\usepackage[usenames]{color}
\usepackage{tabu}
\usepackage{multirow}
\usepackage{booktabs}
\usepackage{algpseudocode}
\usepackage{algorithm}

\newtheorem{theorem}{Theorem}[section]

\newtheorem{lemma}[theorem]{Lemma}

\newtheorem{definition}[theorem]{Definition}

\newcommand{\Pj}{ { {\tt Pj}_{\sigma^T_i} } }

\def\CC{\mathbb C}
\def\RR{\mathbb R}
\def\ZZ{\mathbb Z}

\def\cA{\mathcal A}
\def\cB{\mathcal B}

\def\cD{\mathcal D}

\def\cI{\mathcal I}

\def\cL{\mathcal L}

\def\cN{\mathcal N}
\def\cS{\mathcal S}

\def\cX{\mathcal X}

\def\bd{\mathbf d}

\def\bzero{\mathbf 0}

\newcommand{\hide}[1]{}
\newcommand{\raf}[1]{(\ref{#1})}

\newcommand{\cP}{\ensuremath{\mathcal{P}}}

\newcommand{\cR}{\ensuremath{\mathcal{R}}}
\newcommand{\cV}{\ensuremath{\mathcal{V}}}

\newcommand{\T}{\ensuremath{\mathcal T}}

\newcommand{\OPT}{\ensuremath{\textsc{Opt}}}

\newcommand{\argmin}{\ensuremath{\mathrm{argmin}}}

\newcommand{\argmax}{\operatorname{argmax}}

\DeclareCaptionType{copyrightbox}

\title{Truthful Mechanisms for Combinatorial \\AC Electric Power Allocation}

\numberofauthors{1} 
 
\author{\alignauthor Chi-Kin Chau,\quad Khaled Elbassioni,\quad Majid Khonji\\
             \affaddr{Masdar Institute of Science and Technology, Abu Dhabi, UAE}\\
       \email{\{ckchau,kelbassioni,mkhonji\}@masdar.ac.ae}  
}

\begin{document}

\maketitle

\begin{abstract}
Traditional studies of combinatorial auctions often only consider linear constraints (by which the demands for certain goods are limited by the corresponding supplies). The rise of smart grid presents a new class of auctions, characterized by quadratic constraints. Yu and Chau [AAMAS 13'] introduced the {\em complex-demand knapsack problem}, in which the demands are complex-valued and the capacity of supplies is described by the magnitude of total complex-valued demand. This naturally captures the power constraints in AC electric systems. In this paper, we provide a more complete study and generalize the problem to the multi-minded version, beyond the previously known $\frac{1}{2}$-approximation algorithm for only a subclass of the problem. More precisely, we give a truthful PTAS for the case $\phi\in[0,\frac{\pi}{2}-\delta]$, and a truthful FPTAS, which {\it fully} optimizes the objective function but violates the capacity constraint by at most $(1+\epsilon)$, for the case $\phi\in(\frac{\pi}{2},\pi-\delta]$, where $\phi$ is the maximum angle between any two complex-valued demands and $\epsilon,\delta>0$ are arbitrarily small constants. 
\end{abstract}



\section{Introduction} \label{sec:intro} 

Auctions are vital venues for the interactions of multi-agent systems, and their computational efficiency is critical for agent-based automation. Nonetheless, many practical auction problems are combinatorial in nature, requiring carefully designed time-efficient approximation algorithms. Although there have been decades of research in approximating combinatorial auction problems, traditional studies of combinatorial auctions often only consider linear constraints. Namely, the demands for certain goods are limited by the respective supplies, described by linear constraints.

Recently, the rise of smart grid presents a new class of auction problems. One of the salient characteristics is the presence of periodic time-varying entities (e.g., power, voltage, current) in AC (alternating current) electric systems, which are often expressed in terms of complex numbers\footnote{\scriptsize In the common terminology of power systems \cite{GS94power}, the real part of complex-valued power is known as {\em active} power, the imaginary part is {\em reactive} power, whereas the magnitude is {\em apparent} power. Electric equipment has various active and reactive power requirements, whereas power transmission systems and generators are restricted by the supported apparent power.}. In AC electric systems, it is natural to use a quadratic constraint, namely the magnitude of complex numbers, to describe the system capacity. Yu and Chau  \cite{YC13CKS} introduced the {\em complex-demand knapsack problem} (CKP) to model a one-shot auction for combinatorial AC electric power allocation, which is a quadratic programming variant of the classical knapsack problem.

Furthermore, future smart grids will be automated by agents representing individual users. Hence, one might expect these agents to be self-interested and may untruthfully report their utilities or demands. This motivates us to consider truthful (aka. incentive-compatible) approximation mechanisms, in which it is in the best interest of the agents to report their true parameters. In \cite{YC13CKS} a monotone $\frac{1}{2}$-approximation algorithm that induces a deterministic truthful mechanism was devised for the complex-demand knapsack problem, which however assumes that all complex-valued demands lie in the positive quadrant. 

In this paper, we provide a complete study and generalize the complex-demand knapsack problem to the multi-minded version, beyond the previously known $\frac{1}{2}$-approximation algorithm. More precisely, we consider the problem under the framework of (bi-criteria) $(\alpha,\beta)$-approximation algorithms, which compute a feasible solution with objective function within a factor of $\alpha$ of optimal, but may violate the capacity constraint by a factor of at most $\beta$. 
We give a (deterministic) truthful $(1-\epsilon,1)$-approximation algorithm for the case $\phi\in[0,\frac{\pi}{2}-\delta]$, and a truthful $(1,1+\epsilon)$-approximation for the case $\phi\in(\frac{\pi}{2},\pi-\delta]$, where $\phi$ is the maximum angle between any two complex-valued demands and $\epsilon,\delta>0$ are arbitrarily small constants. 
Moreover, the running time in the latter case is polynomial in $1/(\epsilon\tan\delta)$ (the so-called {\it FPTAS with resource augmentation}). We complement these results by showing that, unless P=NP, neither a PTAS can exist for the latter case nor any bi-criteria approximation algorithm with polynomial guarantees for the case when $\phi$ is arbitrarily close to $\pi$. Our results completely settle the open questions in \cite{YC13CKS}.

%
%

Because of the paucity of space, some proofs are deferred to the extended paper. 

\section{Related Work} \label{sec:related} 

Linear combinatorial auctions can be formulated as variants of the classical knapsack problem \cite{CK00,KPP10book,FC84alg}. Notably, these include the {\em one-dimensional knapsack problem} ({\sc 1DKP}) where each indivisible item has only one single copy, and its multi-dimensional generalization, the {\em $m$-dimensional knapsack problem} ({\sc $m$DKP}).
There is an FPTAS for {\sc 1DKP} \cite{KPP10book}. 

In mechanism design setting, where each customer may untruthfully report her utility and demand, it is desirable to design {\it truthful} or  {\it incentive-compatible} approximation mechanisms, in which it is in the best interest of each customer to reveal her true utility and demand \cite{DN07}.
In the so-called {\it single-minded case}, a {\it monotone} procedure can guarantee incentive compatibility \cite{NRTV07}.  While the straightforward FPTAS for  {\sc 1DKP} is not monotone, since the scaling factor involves the maximum item value, \cite{BKV05KS} gave a monotone FPTAS, by performing the same procedure with a series of different scaling factors irrelevant to the item values and taking the best solution out of them. Hence, {\sc 1DKP} admits an truthful FPTAS. More recently, a truthful PTAS, based on dynamic programming and the notion of the so-called {\it maximal-in-range} mechanism, was given in \cite{DN10} for the {\it multi-minded} case.      

As to {\sc $m$DKP} with $m\geq 2$, a PTAS is given in \cite{FC84alg} based on the integer programming formulation, but it is not evident to see whether it is monotone.     
On the other hand, {\sc 2DKP} is already inapproximable by an FPTAS unless P = NP, by a reduction from {\sc equipartition} \cite{KPP10book}.  Very recently, \cite{KTV13} gave a truthful FPTAS with $(1+\epsilon)$-violation for multi-unit combinatorial auctions with a constant number of distinct goods (including {\sc $m$DKP}), and its generalization to the multiple-choice version, when $m$ is fixed. Their technique is based on applying the VCG-mechanism to a rounded problem. Based on the PTAS for the multi-minded {\sc 1DKP} developed in \cite{DN10}, they also obtained a truthful PTAS for the  multiple-choice multidimensional knapsack problem.  

In contrast, non-linear combinatorial auctions were explored to a little extent.  Yu and Chau  \cite{YC13CKS} introduced complex-demand knapsack problem, which models auctions with a quadratic constraint.

\section{Problem Definitions and Notations}\label{sec:model}

\subsection{Complex-demand Knapsack Problem}

We adopt the notations from \cite{YC13CKS}. 
Our study concerns power allocation under a capacity constraint on the magnitude of the total satisfiable demand (i.e., apparent power). Throughout this paper, we sometimes denote $\nu^{\rm R} \triangleq {\rm Re}(\nu)$ as the real part and $\nu^{\rm I} \triangleq {\rm Im}(\nu)$ as the imaginary part of a given complex number $\nu$. We also interchangeably denote a complex number by a 2D-vector as well as a point in the complex plane. $|\nu|$ denotes the magnitude of $\nu$.

We define the single-minded complex-demand knapsack problem ({\sc CKP}) with a set $\cN=[n]\triangleq\{1,\ldots,n\}$ of users as follows: 
\begin{eqnarray}
\textsc{(CKP)} \qquad& \displaystyle \max_{x_k \in \{0, 1 \}} \sum_{k\in\cN}u_k x_k \label{CKP}\\
\text{subject to}\qquad & \displaystyle \Big|\sum_{k\in \cN}d_k x_k\Big| \le C. \label{C1}
\end{eqnarray}
where $d_k = d_k^{\rm R} + {\bf i} d_k^{\rm I} \in\CC$ is the {\em complex-valued} demand of power for the $k$-th user, $C \in\RR_+$ is a real-valued capacity of total satisfiable demand in apparent power. 
Evidently, {\sc CKP} is also NP-complete, because the classical 1-dimensional knapsack problem ({\sc 1DKP}) is a special case. 

We note that the problem is invariant, when the arguments of all demands are shifted by the same angle. 
Without loss of generality, we assume that one of the demands, say $d_1,$ is aligned along the positive real axis, and define a class of sub-problems for {\sc CKP}, by restricting the maximum phase angle (i.e., the argument) that any other demand makes with $d_1$. In particular, we will write {\sc CKP}$[\phi_1,\phi_2]$ for the restriction of problem {\sc CKP} subject to $\phi_1 \le \max_{k \in \cN}{\rm arg}(d_k)$ $\le \phi_2$, where ${\rm arg}(d_k)\in [0,\pi]$ is the angle that $d_k$ makes with $d_1$. We remark that in realistic settings of power systems, the active power demand is positive (i.e., $d_k^{\rm R} \ge 0$), but the power factor (i.e., $\frac{d^{\rm R}_k}{|d_k|}$) is bounded by a certain threshold \cite{NEC}, which is equivalent to restricting the argument of complex-valued demands. 

From the computational point of view, we will need to specify how the inputs are described. Throughout the paper we will assume that each of the demands is given by its real and imaginary components, represented as rational numbers.

\subsection{Non-single-minded Complex Knapsack Problem}
In this paper, we extend the single-minded {\textsc CKP} to general {\it non-single-minded} version, and then we apply the well-known {\it VCG-mechanism}, or equivalently the framework of {\it maximal-in-range} mechanisms \cite{NR07}. The non-single-minded version is defined as follows. As above we assume a set $\cN$ of $n$ users: user $k$ has a valuation function $v_k:\cD\to\RR_+$ over a (possibly infinite) set of demands $\cD\subseteq\CC$. We assume that $\bzero\in\cD$, $v_k(\bzero)=0$ for all $k\in\cN$, and w.l.o.g., $|d|\le C$ for all $d\in\cD$. We further assume that each $v_k$ is {\it monotone} with respect to a partial order "$\preceq$" defined on the elements of $\CC$ as follows: for $d,f\in\CC$, $d\succeq f$ if and only if
{\small
$$
|d^{\rm R}|\ge |f^{\rm R}|,|d^{\rm I}|\ge |f^{\rm I}|,{\rm sgn}(d^{\rm R}) = {\rm sgn}(f^{\rm R}),{\rm sgn}(d^{\rm I}) = {\rm sgn}(f^{\rm I}).
$$
}
(We assume $\bzero\preceq d$ for all $d\in\cD$.) Then for all $k\in\cN$, the monotonicity of $v_k(\cdot)$ means that $v_k(d)\ge v_k(f)$ whenever $d\succeq f$.

The non-single-minded problem can be described by the following program (in the variables $d_k$): 
\begin{eqnarray}
\textsc{(NsmCKP)}&\displaystyle \max  \sum_{k\in\cN}v_k (d_k) \label{(CV-nsm-KS)}\\
\text{s.t.} & \displaystyle (\sum_{k\in\cN} d_k^{\rm R})^2+ (\sum_{k\in\cN}d_k^{\rm I})^2\le C^2 \label{nsm-CV1}\\
& d_k\in\cD \text{ for all }k\in \cN. \label{nsm-CV2}
\end{eqnarray}

Of particular interest is the {\it multi-minded} version of the problem ({\sc MultiCKP}), defined as follows. Each user $k\in\cN$ is interested only in a {\it polynomial-size} subset of demands $D_k\subseteq \cD$ and declares her valuation only over this set. Note that the multi-minded problem can be modeled in the form \textsc{(NsmCKP)} by assuming w.l.o.g. that $\bzero\in D_k$, for each user $k\in\cN$, and defining the valuation function $v_k:\cD\to\RR_+$ as follows: 
\begin{equation}\label{mm-val}
v_k(d)=\max_{d_k\in D_k}\{v_k(d_k):~d_k\preceq d \}.
\end{equation}
We shall assume that the demand set of each user lies completely in one of the quadrants, namely, either $d^{\rm R}\ge 0$ for all $d\in D_k$, or $d^{\rm R}< 0$ for all $d\in D_k$.  
Note that the single-minded version (which is \textsc{CKP}) is special case, where $|D_k|=1$ for all $k$.   

We will write {\sc MultiCKP}$[\phi_1,\phi_2]$ for the restriction of the problem  subject to $\phi_1 \le \phi \le \phi_2$ for all $d \in \cD$ where $\phi\triangleq\max_{d\in\cD}{\rm arg}(d)$ (and as before we assume ${\rm arg}(d)\ge 0$). 

\subsection{Multiple-choice Multidimensional Knapsack Problem}\label{MCMDKS-sec}
To design truthful mechanisms for \textsc{NsmCKP}, it will be useful to consider the {\it multiple-choice multidimensional knapsack} problem (\textsc{Multi-$m$DKP}) defined as follows, where we assume more generally that $\cD\subseteq\RR_+^m$ and a {\it capacity vector} $c\in\RR_+^m$ is given. As before, a valuation function for each user $k$ is given by \raf{mm-val}. 
An {\it allocation} is given by an assignment of a demand $d_k=(d_k^1,...,d_k^m)\in\cD$ for each user $k$, so as to satisfy the $m$-dimensional capacity constraint $\sum_{k\in\cN}d_k\le c$. The objective is to find an allocation $\bd=(d_1,\ldots,d_n)\in\cD^n$ so as to maximize the sum of the valuations $\sum_{k\in\cN}v_k(d_k)$. The problem can be described by the following program:

\begin{eqnarray}
\textsc{(Multi-$m$DKP)} & \displaystyle \max\sum_{k\in\cN}v_k (d_k)& \label{(mCmD-KS)}\\
\text{ s.t. } & \displaystyle \sum_{k\in\cN} d_k &\le  c \label{mCmD-1}\\
& d_k\in D_k,&~~~~\forall k\in\cN.  \label{mCmD-3}
\end{eqnarray}

\subsection{Approximation Algorithms}

We present an explicit definition of approximation algorithms for our problem.
Given a feasible allocation $\bd=(d_1,\ldots,d_n) \in \cD^n$ satisfying \raf{nsm-CV1}, we write $v(\bd)\triangleq\sum_{k\in\cN}v_k(d_k)$.
Let $ \bd^\ast$ be an optimal allocation of \textsc{NsmCKP} (or (\textsc{MultiCKP})) and $\OPT \triangleq v(\bd^\ast)$ be the corresponding total valuation. We are interested in polynomial time algorithms that output an allocation that is within a factor $\alpha$ of the optimum total valuation, but may violate the capacity constraint by at most a factor of $\beta$:  
\begin{definition}
For $\alpha\in(0,1]$ and $\beta\ge 1$, a bi-criteria $(\alpha,\beta)$-approximation to \textsc{NsmCKP} is an allocation $(d_k)_{k} \in \cD^n$ satisfying 
\begin{eqnarray}
& \displaystyle \Big|\sum_{k\in\cN}d_k\Big| \le \beta \cdot C \label{C1'}\\
\text{such that}\qquad & \displaystyle \sum_{k\in\cN}v_k(d_k) \ge \alpha \cdot \OPT.
\end{eqnarray}
Similarly we define an $(\alpha,\beta)$-approximation to \textsc{MultiCKP}. 
\end{definition}
In particular, a {\em polynomial-time approximation scheme} (PTAS) is a $(1-\epsilon,1)$-approximation algorithm for any $\epsilon>0$.  The running time of a PTAS is polynomial in the input size for every fixed $\epsilon$, but the exponent of the polynomial may depend on $1/\epsilon$.  
An even stronger notion is a {\em fully polynomial-time approximation scheme} (FPTAS), which requires the running time to be polynomial in both input size and $1/\epsilon$. 
In this paper, we are interested in an FPTAS in the {\it resource augmentation model}, which is a $(1, 1+\epsilon)$-approximation algorithm for any $\epsilon>0$, with the running time being polynomial in the input size and $1/\epsilon$. We will refer to this as a $(1,1+\epsilon)$-FPTAS.

\subsection{Truthful Mechanisms}
This section follows the terminology of \cite{NRTV07}.
We define truthful (aka. incentive-compatible) approximation mechanisms for our problem. We denote by $\cX\subseteq\cD^n$ the set of {\it feasible allocations} in our problem (\textsc{NsmCKP} or \textsc{Multi-mDKP}). 

\begin{definition}[Mechanisms]\label{d3}
Let $\cV\triangleq\cV_1\times\cdots\times\cV_n$, where $\cV_k$ is the set of all possible valuations of agent $k$. 
A mechanism $(\cA,\cP)$ is defined by an allocation rule $\cA:\cV\to\cX$ and a payment rule $\cP:\cV\to\RR^n_+$. We assume that the utility of player $k$, under the mechanism, when it receives the vector of bids $v\triangleq(v_1,\ldots,v_n)\in\cV$, is defined as $U_k(v)\triangleq\bar v_k(d_k(v))-p_k(v)$, where $\cA(v)=(d_1(v),\ldots,d_n(v)),$ and $\cP(v)=(p_1(v),\ldots,p_n(v))$ and $\bar v_k$ denotes the true valuation of player $k$.   
\end{definition}
Namely, a mechanism defines an allocation rule and payment scheme, and the utility of a player is defined as the difference between her valuation over her allocated demand and her payment.

\begin{definition}[Truthful Mechanisms]\label{d4}
A mechanism is said to be {\it truthful} if for all $k$ and all $v_k\in\cV_k$, and $v_{-k}\in\cV_{-k}$, it guarantees that $U_k(\bar v_k,v_{-k})\geq U_k(v_k,v_{-k})$. 
\end{definition}
Namely, the utility of any player is maximized, when she reports the true valuation. 
\begin{definition}[Social Efficiency]\label{d41}
A mechanism is said to be {\it $\alpha$-socially efficient} if for any $v\in\cV$, it returns an allocation $\bd\in\cX$ such that the  total valuation (also called {\it social welfare}) obtained is at least an $\alpha$-fraction of the optimum: $v(\bd)\ge\alpha \cdot 
\OPT$. 
\end{definition}
As in \cite{NR07,DN10,KTV13}, our truthful mechanisms are based on using {\it VCG payments} with {\it Maximal-in-Range} (MIR) allocation rules:
\begin{definition}[MIR]\label{d5}
An allocation rule $\cA:\cV\to\cX$ is an MIR, if there is a range $\cR\subseteq\cX$, such that for any $v\in\cV$, $\cA(v)\in\argmax_{\bd\in\cR}v(\bd)$.
\end{definition}
Namely, $\cA$ is an MIR if it maximizes the social welfare over a fixed ({\it declaration-independent}) range $\cR$ of feasible allocations. It is well-known (and also easy to prove by a VCG-based argument) that an MIR, combined with VCG payments (computed with respect to range $\cR$), yields a truthful mechanism. If, additionally, the range $\cR$ satisfies: $\max_{\bd\in\cR}v(\bd)\ge\alpha\cdot \max_{\bd\in\cX}v(\bd)$, then such a mechanism is also $\alpha$--socially efficient.    

Finally a mechanism is {\it computationally efficient} if it can be implemented in polynomial time (in the size of the input).

\section{A Truthful PTAS for {\sc MultiCKP}$[0,\frac{\pi}{2}-\delta]$}
Problem \textsc{Multi-$m$DKP} was shown in \cite{KTV13} to have a $(1-\epsilon)$-socially efficient truthful PTAS in the setting of {\it multi-unit auctions with a few distinct goods}, based on generalizing the result for the case $m=1$ in \cite{DN10}. We explain this result first in our setting, and then use it in Sections~\ref{sec:ptas} and~ \ref{sec:truthful-ptas} to derive a truthful PTAS for {\sc MultiCKP}$[0,\frac{\pi}{2}-\delta]$. We remark that, without the truthfulness requirement, our PTAS works even for $\delta=0$. However, we are only able to make it truthful for any given, but arbitrarily small, constant $\delta>0$.  Removing this technical assumption is an interesting open question.    

\subsection{A Truthful PTAS for \textsc{Multi-$m$DKP}}
\label{sec:Tm-mC-KS}

Let $c=(c^1,\ldots,c^m)$ be the capacity vector, and for any $d\in\cD\subseteq\RR^m_+$, write $d_k=(d^1_k,\ldots,d^m_k)$. For any subset of users $N\subseteq\cN$ and a partial selection of demands $\bar\bd=(d_k\in\cD:~k\in N)$, such that $\sum_{k\in N}d_k\le c$, define the vector $b_{N,\bar\bd}=(b_{N,\bar\bd}^1,\ldots,b_{N,\bar\bd}^m)\in\RR^m_+$ as follows
\begin{equation}\label{bdT}
b_{N,\bar\bd}^i=\frac{c^i-\sum_{k\in N}d_k^i}{(n-{|N|})^2}.
\end{equation}
Following \cite{NR07,KTV13}, we consider a restricted range of allocations defined as follows: 
\begin{equation}\label{range}
\cS\triangleq\bigcup_{\stackrel{N\subseteq\cN,~\bar\bd=(d_k:~k\in N):~|N|\le\frac{m}{\epsilon},}{d_k\in\cD~\forall k\in N}}\cS_{N,\bar\bd},
\end{equation}
where, for a set $N\subseteq\cN$ and a partial selection of demands $\bar\bd=(\bar d_k\in \cD:~k\in N)$,
{\small
\begin{align*}
\cS_{N,\bar\bd}&\triangleq\Big\{(d_1,\ldots, d_n)\in\cD^n~|~\sum_{k\in\cN}d_k\le c, d_k=\bar d_k\ \forall k\in N,\\
&\forall k\not\in N~\forall i\in[m]~\exists r_k^i\in\ZZ_+ \mbox{\ s.t.\ }  d_k^i=r_k^i\cdot b_{N,\bar\bd}^i\\
&\text{ and } \sum_{k\not\in N}r_k^i\le (n-|N|)^2~\Big\}.
\end{align*}}
\hspace{-0.05in}Note that the range $\cS$ {\it does not} depend on the declarations $D_1,\ldots,D_n$. The following two lemmas establish that the range $\cS$ is a good approximation of the set of all feasible allocations and that it can be optimized over in polynomial time. The first lemma is essentially a generalization of a similar one for multi-unit auctions in \cite{DN10}, with the simplifying difference that we do not insist here on demands to be integral. The second lemma is also a generalization of a similar result in \cite{DN10}, which was stated for the multi-unit auctions with a few distinct goods in \cite{KTV13}. 
\begin{lemma}[\cite{DN10}]\label{l1-}
$\max_{\bd\in\cS}v(\bd)\ge(1-\epsilon)\OPT.$
\end{lemma}

\begin{lemma}[\cite{DN10,KTV13}]\label{l2-}
We can find $\bd^*\in\argmax_{\bd\in\cS}v(\bd)$ using dynamic programming in time $\left|\bigcup_{k}D_k\right|^{O(m/\epsilon)}$.
\end{lemma}

It follows that an allocation rule defined as an MIR over range $\cS$ yields a $(1-\epsilon)$-socially efficient truthful mechanism for \textsc{Multi-$m$DKP}.


\subsection{A PTAS for {\sc MultiCKP}$[0,\frac{\pi}{2}]$} \label{sec:ptas}
We now apply the result in the previous section to the multi-minded complex-demand knapsack problem, when all agents are restricted to report their demands in the positive quadrant. We begin first by presenting a PTAS without strategic considerations; then is shown in the next section how to use this PTAS within the aforementioned framework of MIR's to obtain a truthful mechanism.

In this section  we assume that ${\rm arg}(d) \le \frac{\pi}{2}$, that is, $d^{\rm R} \ge 0$ and $d^{\rm I} \ge 0$ for all $d\in\cD$. 
As we shall see in Section~\ref{sec:tbp}, it is possible to get a $(1,1+\epsilon)$-approximation by a reduction to the \textsc{Multi-$2$DKP} problem. We note further that although there is a PTAS for {\sc $m$DKP} with constant $m$ \cite{FC84alg}, such a PTAS cannot be directly applied to {\sc MultiCKP}$[0,\frac{\pi}{2}]$ by polygonizing the circular feasible region for {\sc MultiCKP}$[0,\frac{\pi}{2}]$, because one can show that such an approximation ratio is at least a constant factor. This is the case, for instance, if the optimal solution consists of a few large (in magnitude) demands together with many small demands, and it is not clear at what level of accuracy we should polygonize the region to be able to capture these small demands. To overcome this difficulty, we first guess the large demands, then we construct a grid (or a lattice) on the remaining part of the circular region, defining a polygonal region in which we try to pack the maximum-utility set of demands. The latter problem is easily seen to be a special case of the {\sc Multi-$m$DKP} problem. The main challenge is to choose the granularity of the grid small enough to well-approximate the optimal, but also large enough so that the number of sides of the polygon, and hence $m$ is a constant only depending on $1/\epsilon$.  

Without loss of generality, we assume $\epsilon<\frac{1}{4}$ where $\frac{1}{\epsilon}\in\ZZ_+$.
For an integer $i\in\ZZ_+$, let $\cL_1(i)$ and $\cL_2(i)$, respectively, denote the sets of all horizontal and all vertical lines in the complex plane that are at (non-negative) distances, form the real and imaginary axes, which are integer multiples of $\frac{C}{2^i}$, that is,
\begin{eqnarray*}
\cL_1(i)\triangleq\{x+{\bf i}y\in \CC~|~x=\frac{\lambda C}{2^i},~\lambda\in \ZZ_+\},\\
\cL_2(i)\triangleq\{x+{\bf i}y\in \CC~|~y=\frac{\lambda C}{2^i},~\lambda\in \ZZ_+\},
\end{eqnarray*}

Given a feasible set of vectors $T\subseteq \cD$ to \textsc{MultiCKP$[0,\frac{\pi}{2}]$} (that is, $\left|\sum_{d\in T}d\right|\le C$), define $d_T\triangleq\sum_{d \in T} d$, and let
\begin{equation}\label{wT}
w^{\rm I}_T \triangleq \sqrt{C^2 -{\rm Re}(d_T)^2} - {\rm Im}(d_T),~ w^{\rm R}_T \triangleq \sqrt{C^2 -{\rm Im}(d_T)^2} - {\rm Re}(d_T).
\end{equation}
Let $\rho_1(T)$ and $\rho_2(T)$ be the smallest integers such that 
$$
\frac{C}{2^{\rho_1(T)}}\le\frac{\epsilon w_T^{\rm R} }{4} \text{ and } \frac{C}{2^{\rho_2(T)}}\le\frac{\epsilon w_T^{\rm I}}{4}. 
$$
The set of lines in $\cL_1(\rho_1(T))\cup\cL_2(\rho_2(T))$ define a grid on the feasible region at ``vertical and horizontal levels'' $\rho_1(T)$ and  $\rho_2(T)$, respectively.
 
Let $\lambda_1(T)$ and $\lambda_2(T)$ be the largest integers such that
$$
d_T^{\rm R}\ge \frac{\lambda_1(T) C}{2^{\rho_1(T)}}~\text{ and }~d_T^{\rm I}\ge \frac{\lambda_2(T) C}{2^{\rho_2(T)}},
$$
and $z_T\in\CC$ be the intersection of the two lines corresponding to $\lambda_1(T)$ and $\lambda_2(T)$: 
$$
z_T\triangleq\{x+{\bf i}y\in \CC~|~x=\frac{\lambda_1(T) C}{2^{\rho_1(T)}}\}\cap\{x+{\bf i}y\in \CC~|~y=\frac{\lambda_2(T) C}{2^{\rho_2(T)}}\}.
$$
Given $z_T$, we define four points in the complex plane $({\pi'}_T^1,\pi_T^1,\pi_T^2,{\pi'}_T^2)$ such that
{\small
\begin{eqnarray*}
& {\pi'}_T^1 = \Big(0, \sqrt{C^2 -{\rm Re}(z_T)^2} \Big), ~ \pi_T^1 = \Big({\rm Re}(z_T),\sqrt{C^2 -{\rm Re}(z_T)^2}\Big), \\
& {\pi'}_T^2 = \Big(\sqrt{C^2 -{\rm Im}(z_T)^2}, 0\Big), ~ \pi_T^2 = \Big(\sqrt{C^2 -{\rm Im}(z_T)^2},{\rm Im}(z_T)\Big).
\end{eqnarray*}}
Let $\cR_T$ be the part of the feasible region dominating $z_T$:
\begin{equation}
\cR_T\triangleq\{x+{\bf i}y\in\CC~:~|x+{\bf i}y|\leq C,~x\ge {\rm Re}(z_T), y\ge {\rm Im}(z_T)\},
\end{equation}
and $P_T(\epsilon)$ be the set of intersection points\footnote{For simplicity of presentation, we will ignore the issue of finite precision needed to represent intermediate calculations (such as the square roots above, or the intersection points of the lines of the gird with the boundary of the circle).}
between the grid lines in $\cL_1(\rho_1(T))\cup\cL_2(\rho_2(T))$ and the boundary of $\cR_T$:
$$
P_T(\epsilon)\triangleq\{z\in\cR_T~:~|z|= C\}\cap(\cL_1(\rho_1(T))\cup\cL_2(\rho_2(T))).
$$
The convex hull of the set of points $P_T(\epsilon)\cup\{{\pi'}_T^1,\pi_T^1,\pi_T^2,{\pi'}_T^2, \bzero\}$ defines a polygonized region, which we denote by $\cP_T(\epsilon)$ and its size (number of sides) by $m_T(\epsilon)$ (see Fig.~\ref{f1} for an illustration).   

\begin{lemma}\label{l-size}
$ m_T(\epsilon) \leq \frac{18}{\epsilon}+3$. 
\end{lemma}
\begin{figure}[!htb]
	\centering
	
	\hspace{-30pt}
	\begin{subfigure}{.5\textwidth}  
		\centering
		\includegraphics[scale=0.7]{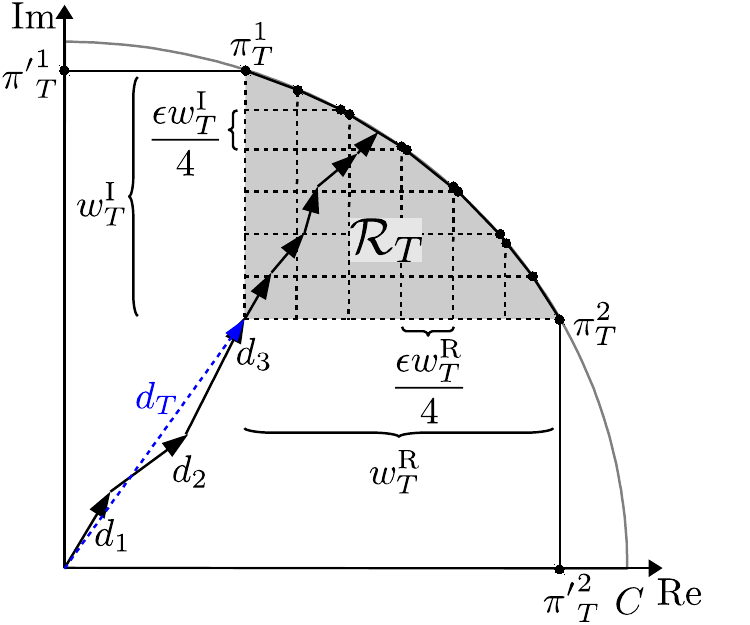}
		\caption{We illustrate the region $\cR_T$ by the shaded area and $P_T(\epsilon)$ by the black dots.}
		\label{f1}
		\end{subfigure}
		\hspace{5pt}
	\begin{subfigure}{.5\textwidth}  
		\centering
		\includegraphics[scale=0.7]{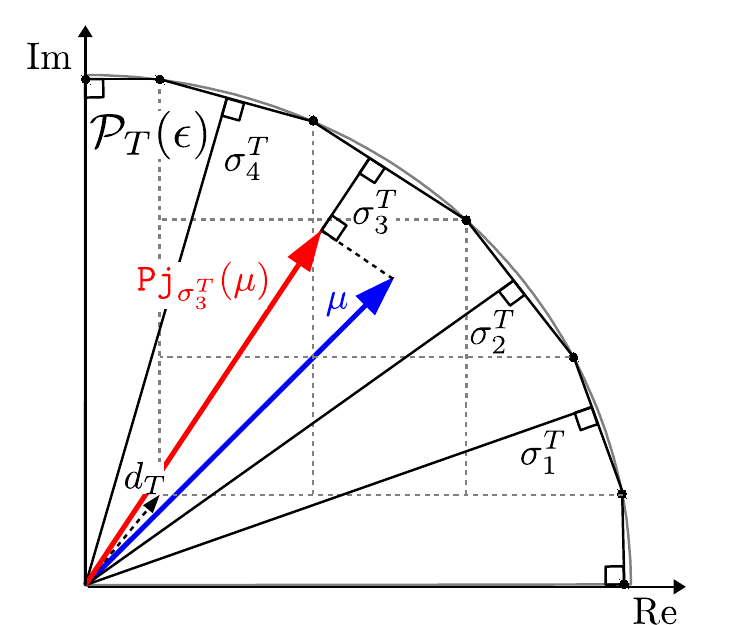}

		\caption{Each in $\{\sigma_T^i\}$ is a vector (starting at the origin) perpendicular to each boundary edge of $\cP_T(\epsilon)$. }

		\label{f2}
	\end{subfigure}
\end{figure}
	%

\begin{definition}
Consider a subset of users $N\subseteq \cN$ and a feasible set $T\triangleq\{\overline d_k:k\in N\}$ to {\sc MultiCKP}$[0,\frac{\pi}{2}]$. We define an approximate problem ({\sc PGZ}$_T$) by polygonizing \textsc{MultiCKP $[0,\frac{\pi}{2}]$}:
\begin{eqnarray*}
\textsc{(PGZ$_T$)} \qquad& \displaystyle \max  \sum_{k\in\cN}v_k (d_k) \label{PGZ}\\
\text{s.t.}\qquad & \displaystyle \sum_{k\in \cN}d_k \in \cP_T(\epsilon) \label{CPGZ}\\
\qquad & d_k=\overline d_k,~~~~\forall k\in N\label{DPGZ}\\ 
\qquad & d_k\in\cD, ~~~~\forall k\in\cN\backslash N.
\end{eqnarray*}
\end{definition}

Given two complex numbers $\mu$ and $\nu$, we denote the projection of $\mu$ on $\nu$ by ${\tt Pj}_\nu(\mu) \triangleq \frac{\nu}{|\nu|^2}(\mu^{\rm R}\nu^{\rm R} +  \mu^{\rm I}\nu^{\rm I})$. Given the convex hull $\cP_T(\epsilon)$, we define a set of $m_T(\epsilon)$ vectors $\{ \sigma_T^i\}$, each of which is perpendicular to each boundary edge of $\cP_T(\epsilon)$ and starting at the origin (see Fig.~\ref{f2} for an illustration).  

\begin{definition}
Consider a subset of users $N\subseteq \cN$ and a feasible set $T\triangleq\{\overline d_k:k\in N\}$ to {\sc MultiCKP}$[0,\frac{\pi}{2}]$. We define a \textsc{Multi-$m$DKP} problem based on $\{ \sigma_T^i\}$:

\begin{eqnarray}
&\textsc{(Multi-$m$DKP$\{ \sigma_T^i\}$)}\qquad \displaystyle \max  \sum_{k\in\cN}v_k (d_{k}) \label{(mDKP)}\\
\text{s.t.}& \displaystyle \sum_{k\in\cN} | \Pj(d_k) | \le  |\sigma_T^i|, \quad \forall i= 1,\ldots, m_T(\epsilon), \label{mCm-1}\\
\qquad & d_k=\overline d_k,~~~~\forall k\in N\label{mCm-3}\\ 
\qquad & d_k\in\cD, ~~~~\forall k\in\cN\backslash N. \label{mCm-2}
\end{eqnarray}
\end{definition}

\begin{lemma}\label{lem-proj}
Given a feasible set $T$ to \textsc{MultiCKP$[0,\frac{\pi}{2}]$},  {\sc PGZ}$_T$ and {\sc Multi-$m$DKP}$\{ \sigma_T^i\}$ are equivalent.
\end{lemma}
Lemma~\ref{lem-proj} follows straightforwardly from the convexity of the polygon $\cP_T(\epsilon)$.

Our PTAS for \textsc{MutliCKP$[0,\frac{\pi}{2}]$} is described in Algorithm {\sc MultiCKP-PTAS}, which enumerates every subset partial selection $T$ of at most $\frac{1}{\epsilon}$ demands, then finds a near optimal allocation for each polygonized region $\cP_T(\epsilon)$ using the PTAS of {\sc Multi-$m$DKP} from Section~\ref{sec:Tm-mC-KS}, which we denote by {\sc Multi-$m$DKP-PTAS}$[\cdot]$.
\begin{algorithm}[!htb]
	\caption{ {\sc MultiCKP-PTAS}$(\{v_k,D_k\}_{k\in\cN},C,\epsilon)$} \label{CKP-PTAS}
\begin{algorithmic}[1]
\Require Users' multi-minded valuations $\{v_k,D_k\}_{k\in \cN}$; capacity $C$; accuracy parameter $\epsilon$
\Ensure $(1-3\epsilon)$-allocation $(\widehat{d}_1,\ldots,\widehat d_n)$ to \textsc{MultiCKP$[0,\frac{\pi}{2}]$}
\State $(\widehat{d}_1,\ldots,\widehat d_n) \leftarrow (\bzero,\ldots,\bzero)$
\For{each subset $N\subseteq \cN$ and each subset $T=(\overline d_k\in D_k:k\in N)$ of size at most $\frac{1}{\epsilon}$ s.t. $\big|\sum_{d\in T}d\big|\le C$}\label{ss1}
  \State Set $d_T \leftarrow \sum_{d\in T}d$, and define the corresponding vectors $\{ \sigma_T^i\}$
  \State Obtain $(d_1,\ldots,d_n) \leftarrow$ {\sc Multi-$m$DKP-PTAS} [{\sc Multi-$m$DKP}$\{ \sigma_T^i\}$] within accuracy $\epsilon$ \label{s1} 
  \If{$\sum_kv_k(\widehat{d}_k) < \sum_kv_k(d_k)$}
  \State $(\widehat{d}_1,\ldots,\widehat d_n) \leftarrow (d_1,\ldots,d_n)$
  \EndIf
\EndFor
\State \Return $(\widehat{d}_1,\ldots,\widehat d_n)$
\end{algorithmic}
\end{algorithm}

\begin{theorem}\label{t2}
For any $\epsilon>0$, Algorithm {\sc MultiCKP-PTAS} finds a $(1-3\epsilon, 1)$-approximation to \textsc{MultiCKP$[0,\frac{\pi}{2}]$}. 
The running time of the algorithm is $\left|\bigcup_k D_k\right|^{O(\frac{1}{\epsilon^2})}$.
\end{theorem}
\begin{proof}
First, the upper bound on the running time of Algorithm {\sc MultiCKP-PTAS} is due to the fact that each of the $\left|\bigcup_k D_k\right|^{O\left (\frac{1}{\epsilon}\right)}$ iterations in line~\ref{ss1} requires invoking the PTAS of {\sc Multi-$m$DKP}, which in turn takes $\left|\bigcup_{k}D_K\right|^{O(m/\epsilon)}$ time, by Lemma~\ref{l2-}, where $m =O(\frac{1}{\epsilon})$. 

The algorithm outputs a feasible allocation by Lemma \ref{lem-proj} and the construction of $\cP_T(\epsilon)$.
To prove the approximation ratio, we show in Lemma~\ref{main-lem} below that, for any optimal (or feasible) allocation $(d_1^*,\ldots,d_n^*)$, we can construct another feasible allocation $(\widetilde d_1,\ldots,\widetilde d_n)$ such that $\sum_{k}v_k(\widetilde d_k)\ge(1-2\epsilon)\sum_kv_k(d_k^*)$ and $(\widetilde d_1,\ldots,\widetilde d_n)$ is feasible to {\sc PGZ$_T$} for some $T$ of size at most $\frac{1}{\epsilon}$. By Lemma \ref{lem-proj}, invoking the PTAS of {\sc Multi-$m$DKP$\{\sigma^i_T\}$} gives  a $(1-\epsilon)$-approximation $(\widehat d_1,\ldots,\widehat d_k)$ to {\sc PGZ$_T$}. Then
{\small
\begin{eqnarray*}
\sum_kv_k( \widehat d_k )\ge (1-\epsilon) \sum_kv_k(\widetilde d_k) \ge  (1-3\epsilon) \OPT.
\end{eqnarray*}}
\hspace{-0.05in}We give an explicit construction of the allocation $(\widetilde d_1,\ldots,\widetilde d_n)$ in Algorithm~\ref{Construct}, thus completing the proof by Lemma~\ref{main-lem}.
\end{proof}

\begin{lemma}\label{main-lem}
Consider a feasible allocation $\bd=(d_1,\ldots,$ $d_n)$ to \textsc{MultiCKP$[0,\frac{\pi}{2}]$}. Then we can find a set $T\subseteq \{d_1,\ldots,$ $d_n\}$ and construct an allocation $\widetilde \bd=(\widetilde d_1,\ldots,\widetilde d_n)$, such that $|T|\le\frac{1}{\epsilon}$ and
$\widetilde\bd$ is a feasible solution to {\sc PGZ}$_T$ and $v(\widetilde\bd) \ge(1-2\epsilon)v(\bd)$.    
\end{lemma}
\begin{lemma}\label{lem-pack}
Consider a set  of demands $S\subseteq\cD$ and $T \subseteq S$, such that
\begin{itemize}
\item $S$ is feasible solution to \textsc{MultiCKP$[0,\frac{\pi}{2}]$}, but $S$ is not a feasible solution to {\sc PGZ}$_T$
\item $d^{\rm R}\le \frac{\epsilon}{4} w^{\rm R}_T$ and  $d^{\rm I}\le \frac{\epsilon}{4} w^{\rm I}_T$, for all $d \in S \backslash T$. 
\end{itemize}
 Then there exists a partition $\{V_1,\ldots, V_h\}$ of $S \backslash T$ such that 
\begin{itemize}
\item 
 either (i) $\sum_{d\in V_j}d^{\rm R}\geq\frac{\epsilon}{4} w^{\rm R}_T$ for all $j={1,\ldots,h}$, 
\item or (ii) $\sum_{d\in V_j}d^{\rm I}\ge\frac{\epsilon}{4}w^{\rm I}_T $ for all $j={1,\ldots,h}$.
\end{itemize} 
where $h\in[\frac{1}{\epsilon}-1,\frac{4}{\epsilon})$.
\end{lemma}
\subsection{Making the PTAS Truthful}\label{sec:truthful-ptas}
We now state our main result for this section.
\begin{theorem}\label{t-TIE-CKP}
For any $\epsilon,\delta>0$ there is a $(1-3\epsilon)$-socially efficient truthful mechanism for {\sc multiCKP}$[0,\frac{\pi}{2}-\delta]$. The running time is $\left|\bigcup_k D_k\right|^{O\left(\frac{\cot^2\frac{\delta}{2}}{\epsilon^2}\right)}$.
\end{theorem}
\begin{proof}
It suffices to define a declaration-independent range $\cS$ of feasible allocations, such that $\max_{\bd\in\cS}v(\bd)\ge(1-3\epsilon)\cdot\OPT$, and we can optimize over $\cS$ in the stated time. 

\medskip

One technical difficulty that arises in this case is that the polygons $\cP_T(\epsilon)$ defined by a guessed initial sets $T$ are not {\it monotone} w.r.t. the set of demands in $T$, that is, if we obtain $T'$ from $T$ by increasing one of th demands from $d_k$ to $d_k'\succ d_k$, then it could be the case that $\cP_T(\epsilon)\not\supseteq \cP_{T'}(\epsilon)$. This implies that the algorithm can be manipulated by a selfish user in $T$ who untruthfully increases his demand to change his allocation by the algorithm and  become a winner. To handle this issue, we will show that the number of possible polygons that arise from such a selfish user, misreporting his true demand set, and can possibly change the outcome, is only a constant in $\epsilon$ and $\delta$. Thus, it would be enough to consider only all such polygons arising from the reported demand set.   
  
Since we assume that $\arg(d)\in[0,\frac{\pi}{2}-\delta]$, for all $d\in \bigcup_kD_k$, we may assume further by performing a rotation that any such vector $d$ satisfies $\arg(d)\in[\frac{\delta}{2},\frac{\pi}{2}-\frac{\delta}{2}]$. 
For convenience, we continue to denote the new demand sets by $D_k$, and redefine the valuation functions in terms of these rotated sets. By this assumption, 
\begin{equation}\label{e-a}
\tan\frac{\delta}{2}\le\frac{d_T^{\rm R}}{d_T^{\rm I}}\le\left(\tan\frac{\delta}{2}\right)^{-1},~~\text{for any $T\subseteq\cD$}.
\end{equation}
We may also assume, by scaling $\epsilon$ by $2/(1+2\cot^2\frac{\delta}{2})$ if necessary, that
\begin{equation}\label{assum}
\epsilon\le\frac{2}{1+2\cot^2\frac{\delta}{2}}.
\end{equation}
For $T\subseteq\cD$, let $G(T)$ be the set of vectors in $\CC$ defined by the union of $\{d_T\}$ and 
\begin{itemize}
\item[(a)] the (component-wise) minimal grid points $z\in\cR_T$, such that $z=\ell_1\cap\ell_2$ for some $\ell_1\in\cL_1(\rho_1(T)+1)$ and $\ell_2\in\cL_2(\rho_2(T)+1)$, and either $\rho_1(\{z\})=\rho_1(T)+1$ {\it or} $\rho_2(\{z\})=\rho_2(T)+1$, but not both; and
\item[(b)] the (component-wise) minimal grid points $z\in\cR_T$, such that $z=\ell_1\cap\ell_2$ for some $\ell_1\in\cL_1(\rho_1(T)+1)$ and $\ell_2\in\cL_2(\rho_2(T)+1)$, and $\rho_1(\{z\})=\rho_1(T)+1$ {\it and} $\rho_2(\{z\})=\rho_2(T)+1$. 
\end{itemize}
Note that $|G(T)|=O(\frac{1}{\epsilon})$.
For convenience of notation, let us fix two subsets $\cD_1,\cD_2\subseteq\cD$.
For $z\in G(T)$, let us denote by $\cS_z(\cD_2)$ the range of feasible allocations defined as in \raf{range} with respect to the \textsc{Multi-$m$DKP} problem with constraints \raf{mCm-1}-\raf{mCm-2}, when 
\begin{itemize}
\item[(I)] $T$ is replaced by $T\cup\{z-d_T\}$ (and hence, $z$ is used to define the polygon $\cP_T(\epsilon)$);
\item[(II)] we add an additional ``dummy'' user $n+1$ to $\cN$ with valuation $v_{n+1}(d)=0$ for all $d\in\cD$, such that the vector $z-d_T$ as allocated to this user; and 
\item[(III)] the set of vectors in $\cN\backslash N$ is chosen from $\cD_2$. 
\end{itemize}
Then we define the range $\cS(\cD_1,\cD_2)$ as the union: $$\cS(\cD_1,\cD_2)\triangleq\bigcup_{T\subseteq\cD_1:~|T|\le\frac{1}{\epsilon}}\left(\bigcup_{z\in G(T)}\cS_{z}(\cD_2)\right).$$ 
By Lemmas \ref{l1-} and \ref{main-lem}, we have $\max_{\bd\in\cS(\cD,\cD)}v(\bd)\ge(1-3\epsilon)\OPT$ (since $d_T\in G(T)$). It remains to argue that we can efficiently  optimize over $\cS(\cD,\cD)$. 
Using Lemma~\ref{cl1-} proved below, we argue that we can solve the optimization problem over $\cS(\cD,\cD)$ assuming that $\cD=\bigcup_{k}D_k$, that is, $\max_{\bd\in \cS(\cD,\cD)}v(\bd)=\max_{\bd\in \cS(\bigcup_{k}D_k,\bigcup_{k}D_k)}v(\bd).$ 
One direction ``$\geq$'' is obvious;
so let us show that $\max_{\bd\in \cS(\cD,\cD)}v(\bd)\leq\max_{\bd\in \cS(\bigcup_{k}D_k,\bigcup_{k}D_k)}v(\bd).$ 

Suppose that $\bd^*=(d_1^*,\ldots,d_n^*)$ is an optimal allocation over $\cS(\cD,\cD)$, but such that $\bd^*\in\cS_{z'}$ for some $z'\in G(T'')$,  $T''\subseteq\cD$, and $T''\not\subseteq \bigcup_kD_k$. Then let us show that there is a set $T\subseteq \bigcup_kD_k$, $z\in G(T)$, and $\widetilde\bd\in\cS_z(\cD)$, such that  $v(\widetilde\bd)=v(\bd^*)$.

Define an allocation $\widetilde\bd$ as follows: Let $N=\{k:~d_k^*\in T''\}$; for each $k\in N$, we choose $\widetilde d_k\in D_k$ such that $\widetilde d_k\preceq d_k^*$ and $v_k(\widetilde d_k)=v_k(d_k^*)$, and we keep $\widetilde d_k=d_k^*$ if $k\not\in N$.
Let us apply the statement of the lemma with $T=\{\widetilde d_k:~k\in N\}$, $T'=T''\cup\{z'-d_{T''}\}$, and $\kappa=\sum_{k:k\not\in N\cup\{n+1\}}d_k^*$. If (i) holds then
$d_T+\kappa\in \cP_T(\epsilon)$ and therefore we have  
\begin{equation}\label{oneDir}
\max_{\bd\in\cS(\cD,\cD)}v(\bd)=\max_{\bd\in \cS(\bigcup_kD_k,\cD)}v(\bd).
\end{equation}
On the other hand,
if (ii) holds, then $\rho_1(T')\in\{\rho_1(T),\rho_1(T)+1\}$ and $\rho_2(T')\in\{\rho_2(T),\rho_2(T)+1\}$. In this case, if $\rho_1(T')=\rho_1(T)$ and $\rho_2(T')=\rho_2(T)$ then $\cP_{T'}(\epsilon)\subseteq     
\cP_{T}(\epsilon)$ (since $d_T\preceq d_{T'}$), in contradiction that (i) does not hold; otherwise, there is a point $z\in G(T)$ such that $z\preceq z'$, $\rho_1(T\cup\{z-d_T\})=\rho_1(T')$ and $\rho_2(T\cup\{z-d_T\})=\rho_2(T')$. Then $z+\kappa\preceq z'+\kappa\in\cP_{T'}(\epsilon)\subseteq\cP_{T\cup\{z-d_T\}}(\epsilon)$, and we get again \raf{oneDir}. 

Finally, we note that
\begin{equation*}
\max_{\bd\in \cS(\bigcup_kD_k,\cD)}v(\bd)=\max_{\bd\in\cS(\bigcup_{k}D_K,\bigcup_kD_k)}v(\bd),
 \end{equation*}as follows from (the proof of) Lemma~\ref{l2-}.

\end{proof}

\begin{lemma}\label{cl1-}
Let $T,T'\subseteq\cD$ be such that $d_T\preceq d_{T'}$. Consider a vector $\kappa\in\CC$ such that $d_{T'}+\kappa\in \cP_{T'}(\epsilon)$. Then either (i) $d_{T}+\kappa\in\cP_{T}(\epsilon)$, or (ii) $\rho_1(T')\le \rho_1(T)+1$ and $\rho_2(T')\le \rho_2(T)+1$.
\end{lemma} 
\begin{proof}
Suppose that $d_{T}+\kappa\not\in\cP_{T}(\epsilon)$. Then it also holds that  $d_{T'}+\kappa\not\in\cP_{T}(\epsilon)$ (since $d_{T'}\succeq d_T)$. This implies that both $d_{T}+\kappa$ and $d_{T'}+\kappa$ lie within the same grid cell at vertical and horizontal levels $\rho_1(T)$ and  $\rho_2(T)$, respectively, and hence $d_{T'}^{\rm R}-d_{T}^{\rm R}\le\frac{C}{2^{\rho_1(T)}}\le\frac{\epsilon w_T^{\rm R}}{4}$ and $d_{T'}^{\rm I}-d_{T}^{\rm I}\le\frac{C}{2^{\rho_2(T)}}\le\frac{\epsilon w_T^{\rm I}}{4}$. 

Form the definition \raf{wT} of $w_{T}^{\rm R}$, we have 
\begin{eqnarray}\label{e1-}
w_{T}^{\rm R}&=&w_{T'}^{\rm R}+d_{T'}^{\rm R}-d_{T}^{\rm R}+\left(d_{T'}^{\rm I}-d_{T}^{\rm I}\right)\left(\frac{d_{T'}^{\rm I}+d_{T}^{\rm I}}{d_{T'}^{\rm R}+w_{T'}^{\rm R}+d_{T}^{\rm R}+w_{T}^{\rm R}}\right)\nonumber\\
&\le&w_{T'}^{\rm R}+\frac{\epsilon w_T^{\rm R}}{4}+\frac{\epsilon w_T^{\rm I}}{4}\left(\frac{d_{T'}^{\rm I}+d_{T}^{\rm I}}{d_{T'}^{\rm R}+d_{T}^{\rm R}}\right)\nonumber\\
&\le &w_{T'}^{\rm R}+\frac{\epsilon w_{T}^{\rm R}}{4}\left(1+\frac{w_T^{\rm I}}{w_T^{\rm R}}\cdot\frac{1}{\tan \frac{\delta}{2}}\right),
\end{eqnarray}
where we use \raf{e-a} in the last inequality. We can upper-bound $w_T^{\rm I}/w_T^{\rm R}$ by $2/\tan\frac{\delta}{2}$ also using \raf{e-a} as follows:
$$
\frac{w_T^{\rm I}}{w_T^{\rm R}}=\frac{\sqrt{1 -\left(\frac{d_T^{\rm I}}{C}\right)^2} + \frac{d_T^{\rm R}}{C}}{ \sqrt{1 -\left(\frac{d_T^{\rm R}}{C}\right)^2} + \frac{d_T^{\rm I}}{C}} \le\frac{1+\frac{d_T^{\rm R}}{C}}{\sqrt{1 -\left(\frac{d_T^{\rm R}}{C}\right)^2} + \frac{d_T^{\rm R}}{C}\tan\frac{\delta}{2}}.
$$
The latter quantity is bounded by $f(1)=2\tan\frac{\delta}{2}$, since the function $f(a)\triangleq\frac{1+a}{\sqrt{1-a^2}+a\tan\frac{\delta}{2}}$ is monotone increasing in $a\in[0,1]$. Using this bound in \raf{e1-} and rearranging terms, we get
\begin{equation}\label{e3-}
w_{T'}^{\rm R}\ge w_{T}^{\rm R}\left(1-\frac{\epsilon}{4}(1+2\cot^{2}\frac{\delta}{2})\right)\ge \frac{1}{2}w_{T}^{\rm R},
\end{equation}
by our assumption~\raf{assum} on $\epsilon$.
From \raf{e3-} and $\frac{\epsilon w_{T'}^{\rm R} }{8}<\frac{C}{2^{\rho_1(T')}}$, and $\frac{C}{2^{\rho_1(T)}}\le\frac{\epsilon w_{T}^{ \rm R}}{4}$, follows that $\rho_1(T')\le \rho_1(T)+1$. Similarly, we have $\rho_2(T')\le \rho_2(T)+1$. 
\end{proof} 

\section{A Truthful FPTAS for {\sc MultiCKP}$[0,\pi \mbox{-} \varepsilon]$}\label{sec:tbp}
As in \cite{KTV13}, the basic idea is to round off the set of possible demands to obtain a range, by which we can optimize over in polynomial time using dynamic programming (to obtain an MIR). 

Let $\theta = \max\{\phi - \frac{\pi}{2},0\}$, where $\phi\triangleq\max_{d\in\cD}{\rm arg}(d)$. We assume that $\tan \theta$ is bounded by an {\it a-priori} known polynomial $P(n)\ge 1$ in $n$, that is {\it independent} of the customers valuations.
We can upper bound the total projections for any feasible allocation $\bd=(d_1,\ldots,d_n)$ of demands as follows:
{\small
\begin{align}
 \sum_{k \in \cN} d_k^{\rm I} \le C, \quad
 \sum_{k \in \cN_- } - d_k^{\rm R}   \le  C \tan \theta, \quad 
 \sum_{k \in \cN_+}  d_k^{\rm R}   &\le C(1+ \tan \theta), \label{eq:ubounds}\nonumber
\end{align}
}
\hspace{-0.05in}where $\cN_+ \triangleq \{ k\in \cN \mid d_k^{\rm R} \ge 0\}$ and $\cN_- \triangleq \{ k \in \cN\mid   d_k^{\rm R} < 0 \}$.
Define $L \triangleq \frac{\epsilon C}{n (P(n)+1)}$, and for $d\in\cD$, define the new rounded demand $\widehat{d}$ as follows:
{\small
\begin{equation}
\widehat d =
\widehat d^{\rm R} + {\bf i} \widehat d^{\rm I} \triangleq 
\left\{\begin{array}{ll}
\left\lceil \frac{d^{\rm R}}{L} \right\rceil \cdot L + {\bf i} \left\lceil \frac{d^{\rm I}}{L} \right\rceil \cdot L, &\text{ if }d^{\rm R}\ge 0,\\[3mm]
\left\lfloor \frac{d^{\rm R}}{L} \right\rfloor \cdot L + {\bf i} \left\lceil \frac{d^{\rm I}}{L} \right\rceil \cdot L, & \text{ otherwise. } 
\end{array}\right.
\label{eq:truc}
\end{equation}}
\hspace{-0.05in}Consider an optimal allocation $\bd^\ast=(d_1^\ast,\ldots,d_n^\ast)$ to \textsc{MultiCKP} $[0,\pi\mbox{-}\varepsilon]$.
Let $\xi_+$ (and $\xi_-$), $\zeta_+$ (and $\zeta_-$) be the respective guessed real and imaginary absolute total projections of the rounded demands in $S^\ast_+\triangleq\{k:d_k^{\rm R}\ge 0\}$ (and $S_-^\ast\triangleq\{k:d_k^{\rm R}< 0\}$).
Then the possible values of $\xi_+, \xi_-, \zeta_+, \zeta_-$ are integral mutiples of $L$ in the following ranges:
{\small
\begin{align*}
 \xi_+ \in {\cal A}_+ & \triangleq \left\{0, L, 2L,\ldots,\left\lceil \frac{C (1 + P(n) )}{L} \right\rceil \cdot L\right\},\\
\xi_- \in {\cal A}_-& \triangleq \left\{0,L, 2L,\ldots, \left\lceil \frac{C \cdot P(n) }{L} \right\rceil\cdot L \right\},\\
\zeta_+,\zeta_-  \in {\cal B}&  \triangleq \left\{0, L, 2L,\ldots,\left\lceil \frac{C}{L} \right\rceil\cdot L\right\}.
\label{eq:grid}
\end{align*}
}
\hspace{-0.05in}Let further $\widehat{\cD}\triangleq \{ \frac{d}{L} \in \cD :~ d^{\rm R}\in\cA_+ \text{ and }d^{\rm I}\in\cB\},$ and note that $|\widehat\cD|=O(\frac{n^2P^3(n)}{\epsilon^2})$.

We first present a  $(1,1+3\epsilon)$-approximation algorithm ({\sc MultiCKP-biFPTAS}) for \textsc{MultiCKP$[0,\pi\mbox{-}\varepsilon]$}.
Let $\cN_+\triangleq \{k \in \cN \mid d^{\rm R} \ge 0~\forall d\in D_k\}$ and $\cN_-\triangleq\{k \in \cN \mid d^{\rm R} < 0~\forall d\in D_k\}$ be the subsets of users with demand sets in the first and second quadrants respectively (recall that we restrict users' demand sets to allow such a partition).

The basic idea of Algorithm {\sc MultiCKP-biFPTAS} is to enumerate the guessed total projections on real and imaginary axes for $S_+^\ast$ and $S_-^\ast$ respectively. 
We then solve two separate {\sc Multi-2DKP} problems (one for each quadrant) to find subsets of demands that satisfy the individual guessed total projections. But since {\sc Multi-2DKP} is generally NP-hard, we need to round the demands to get a problem that can be solved efficiently by dynamic programming. We note that the violation of the optimal solution to the rounded problem w.r.t. to the original problem is small in $\epsilon$. 
\begin{lemma}
For any optimal allocation $\bd^*=(d_1^*,\ldots,d_n^*)$ to \textsc{MultiCKP $[0,\pi\mbox{-}\varepsilon]$}, we have $\big | \sum_{k} \widehat d_k^*\big| \le  ( 1 + 2\epsilon)C$.
 \label{lem-trunc}
\end{lemma}

The next step is to solve the each rounded instance exactly. Assume an arbitrary order on $\cN = \{ 1, ..., n\}$. We define a 3D table, with each entry ${U}(k,c_1, c_2)$ being the maximum utility obtained from a subset of users $\{1,2,\dots,k\} \subseteq \cN$, each  with choosing from $\widehat{\cD}$, that can fit exactly (i.e., satisfies the capacity constraint as an equation) within capacity $c_1$ on the real axis and $c_2$ on the imaginary axis. 
This table can be filled-up by standard dynamic programming; we denote such a program by {\sc Multi-2DKP-Exact}$[\cdot]$. For a user $k\in\cN_-$, we redefine the valuation as
$\bar v_k(d)=v_k(\bar d)$, where, for $d\in\cD$, $\bar d^{\rm R}=-d^{\rm R}$ and $\bar d^{\rm I}=d^{\rm I}$. For a set $F\subseteq\cD$, we write $\bar F$ for the set 
$\{\bar d:~d\in F \}$.

{\small
\begin{algorithm}[!htb]\label{MultiCKP-biFPTAS}
\caption{{\sc MultiCKP-biFPTAS} $( \{v_k,D_k\}_{k \in \cN}, C,\epsilon) $}
\begin{algorithmic}[1]
\Require Users' multi-minded valuations $\{v_k,D_k\}_{k\in \cN}$; capacity $C$; accuracy parameter $\epsilon$
\Ensure {\small $(1,1+3\epsilon)$-allocation $(\widetilde{d}_1,\ldots,\widetilde d_n)$ to \textsc{MultiCKP$[0,\pi\mbox{-}\varepsilon]$}}
\State $({d}_1,\ldots,d_n) \leftarrow (\bzero,\ldots,\bzero)$
\State $\widehat{\cD}_+ \leftarrow \{ \frac{d}{L} \in \cD :~ d^{\rm R}\in\cA_+ \text{ and }d^{\rm I}\in\cB\}$ 
\State $\widehat{\cD}_- \leftarrow \{ \frac{d}{L} \in \cD :~ -d^{\rm R}\in\cA_- \text{ and }d^{\rm I}\in\cB\}$
\ForAll {$\xi_+ \in {\cal A}_+, \xi_- \in {\cal A}_-, \zeta_+, \zeta_- \in {\cal B}$}
\If {$(\xi_+ - \xi_-)^2 + (\zeta_+ + \zeta_-)^2 \le (1+2\epsilon)^2C^2$}\label{cond1}
\State {\small $F_+ \leftarrow \text{\sc Multi-2DKP-Exact}(\{v_k,D_k\}_{k\in\cN_+}, \frac{\xi_+}{L},\frac{\zeta_+}{L},\widehat\cD)$}
\State {\small $F_- \leftarrow \text{\sc Multi-2DKP-Exact}(\{\bar v_k,D_k\}_{k\in \cN_-}, \frac{\xi_-}{L},\frac{\zeta_-}{L},\widehat\cD)$} 
\State $(d_1',\ldots,d_n')\leftarrow F_+ \cup \overline F_-$ 
\If{$\sum_kv_k(d'_k) > \sum_kv_k(d_k)$}
\State $(d_1,\ldots,d_n)\leftarrow  (d_1',\ldots,d_n')$
\EndIf 
\EndIf
\EndFor
\ForAll{$k\in\cN_+$}
\State Choose $\widetilde d_k\in D_k$ s.t. $\widetilde d_k\preceq d_k$ and 
$v_k(d_k)=v_k(\widetilde d_k)$
\EndFor
\State \Return $(\widetilde d_1,\ldots,\widetilde d_n)$
\end{algorithmic}
\end{algorithm}
}

The following lemma states that the allocation returned by {\sc MultiCKP-biFPTAS} does not violate the capacity constraint by more than a factor of $1+3\epsilon$.
\begin{lemma}\label{lem-trunc2}
Let $\widetilde \bd$ be the allocation returned by {\sc MultiCKP-biFPTAS}. Then $|\sum_{k}\widetilde d_k|\le(1+3\epsilon) C$. 
\end{lemma}

\begin{theorem}
For any $\epsilon>0$, there is a truthful for {\sc MultiCKP$[0,\pi \mbox{-} \varepsilon]$}, that returns a $(1,1+3\epsilon)$-approximation. The running time is polynomial in $n$ and $\frac{1}{\epsilon}$.
\end{theorem}
\begin{proof}
We define a declaration-independent range $\cS$ as follows. 
For $\xi_+ \in {\cal A}_+, \xi_- \in {\cal A}_-, \zeta_+, \zeta_- \in {\cal B}$, define
\begin{align*}
\cS_{\xi_+,\xi_+,\zeta-+,\zeta_-}&\triangleq\{\bd=(d_1,\ldots,d_n)\in \widehat\cD_+^n:\\
&~~~~\sum_{k\in\cN_+}d_k^{\rm R}=\xi_+,~\sum_{k\in\cN_+}d_k^{\rm I}=\zeta_+,\\
&~~~~-\sum_{k\in\cN_-}d_k^{\rm R}=\xi_-,~\sum_{k\in\cN_+}d_k^{\rm R}=\zeta_-\}.
\end{align*}
Define further 
\begin{align*}
\cS\triangleq\bigcup_{(\xi_+ - \xi_-)^2 + (\zeta_+ + \zeta_-)^2 \le (1+2\epsilon)^2C^2}\cS_{\xi_+,\xi_+,\zeta-+,\zeta_-}.
\end{align*}
Using Algorithm {\sc MultiCKP-biFPTAS}, we can optimize over $\cS$ in time polynomial in $n$ and $\frac{1}{\epsilon}$.  
Thus, it remains only to argue that the algorithm returns a $(1,1+3\epsilon)$-approximation w.r.t. the original range $\cD^n$. 
To see this, let $d_1^*,\ldots,d_n^*\in\cD$ be the demands allocated in the optimum solution to {\sc MultiCKP}, and $\widetilde d_1,\ldots,\widetilde d_n\in\cD$ be the demands allocated by {\sc MultiCKP-biFPTAS}. Then by Lemma~\ref{lem-trunc}, the truncated optimal allocation $(\widehat d_1^*,\ldots,\widehat d_n^*)$ is feasible with respect to a capacity of $(1+2\epsilon)C$, and thus its projections will satisfy the condition in Step~\ref{cond1} of Algorithm~\ref{MultiCKP-biFPTAS}. It follows that $v(\widetilde\bd) \ge v(\widehat\bd^*)\ge v(\bd^*)=\OPT$, where the second inequality follows from the way we round demands~\raf{eq:truc} and the monotonicity of the valuations. Finally, the fact that the solution returned by{\sc MultiCKP-biFPTAS} violates the capacity constraint by a factor of at most $(1+3\epsilon)$ follows readily from Lemma~\ref{lem-trunc2}. 
\end{proof}

\newpage
\section{Conclusion}

In this paper, we provided truthful mechanisms for an important variant of the knapsack problem with complex-valued demands. We gave a truthful PTAS when all demand sets of users lie in the positive quadrant, and a bi-criteria truthful FPTAS when some of the demand sets can lie in the second quadrant. In the full version of the paper, we show that these are essentially the {\it best possible} results in terms of approximation guarantees, assuming P$\neq$NP. 
\section*{Acknowledgment} We thank Piotr Krysta for informing us about the results in \cite{DN10,KTV13} (presented in Section~\ref{sec:Tm-mC-KS}).

\bibliographystyle{plain}
\bibliography{reference}

\end{document}